\documentclass[10pt]{scrartcl}
\usepackage[latin1]{inputenc}
\usepackage{graphicx}
\usepackage[english]{babel}

\usepackage{amsmath}
\usepackage{amssymb}
\usepackage{amsthm}
\usepackage{amsopn}
\usepackage{pst-all}
\usepackage{subfigure}
\usepackage{natbib}

\renewcommand{\matrix}[1]{\begin{pmatrix}#1\end{pmatrix}}
\newcommand{\Real}{\mathbb{R}}
\newcommand{\e}[1]{\cdot 10^{#1}}
\renewcommand{\det}{\operatorname{det}}

\newtheorem{definition}{\textbf{Definition}}
\newtheorem{proposition}{\textbf{Proposition}}
\newtheorem{lemma}{\textbf{Lemma}}
\newtheorem{theorem}{\textbf{Theorem}}

\title{Searching bifurcations in high-dimensional parameter space via a feedback loop breaking approach}
\author{Steffen Waldherr$^{\ast}$\thanks{$^\ast$Corresponding author. Email: waldherr@ist.uni-stuttgart.de} 
and Frank Allg{\"o}wer\\
Institute for Systems Theory and Automatic Control,\\
Universit{\"a}t Stuttgart, Stuttgart, Germany}

\titlehead{\centering\parbox{0.8\linewidth}{\footnotesize\centering 
This is an electronic version of an article published in International
Journal of Systems Science 40:769--782 (2009);\\ International Journal of Systems Science 
is available online at \texttt{http://journalsonline.tandf.co.uk/}.\\
The publisher's version of this article is available online at \texttt{http://www.informaworld.com/openurl?genre=article\&issn=0020\%2d7721\&volume=40\&issue=7\&spage=769}.}}

\begin{document}
\maketitle
\begin{abstract}
Bifurcations leading to complex dynamical behaviour of non-linear systems
are often encountered when the characteristics of feedback circuits in the system are varied.
In systems with many unknown or varying parameters, it is an interesting, but difficult
problem to find parameter values for which specific bifurcations occur.
In this paper, we develop a loop breaking approach to evaluate the influence of
parameter values on feedback circuit characteristics. 
This approach allows a theoretical classification of feedback circuit characteristics
related to possible bifurcations in the system.
Based on the theoretical results, a numerical algorithm for bifurcation search in a
possibly high-dimensional parameter space is developed.
The application of the proposed algorithm is illustrated by searching for a Hopf bifurcation in a model of
the mitogen activated protein kinase (MAPK) cascade,
which is a classical example for biochemical signal transduction.
\end{abstract}

\section{Introduction}

A frequent challenge in the analysis of non-linear dynamical systems is to find parameter
values for which the system undergoes changes in its dynamical behaviour.
Such changes are directly related to the emergence of complex dynamical behaviour.
Standard cases of complex dynamical behaviour are multistability, i.e.\ the
existence of several stable steady states, limit cycle oscillations, and non-periodic oscillations.

Feedback circuits are the major structural feature in the emergence of 
complex dynamical behaviour. In particular, it can be shown that a positive feedback
circuit in the system is required for multistationarity \citep{KaufmanSou2007},
whereas a negative circuit is typically required for limit cycle
oscillations \citep{Snoussi1998}.
This importance of feedback circuits makes control theory a natural tool
for the analysis of complex dynamical behaviour.

Yet, the main properties of a system's qualitative dynamical behaviour are the location
and stability of equilibrium points. 
Knowledge of these is often also useful when analysing complex dynamical behaviour.
It is well known from dynamical systems theory that two
stable equilibrium points are separated by an invariant repellor,
which contains an unstable equilibrium point in most cases. 
Similarly, stable limit cycle oscillations usually coexist with an unstable equilibrium point.
Also transient behaviour is often governed by the attraction to and repulsion from equilibrium points.
Thus a convenient first step when studying the qualitative behaviour of a dynamical system is
to look at stability properties of equilibrium points.

A classical tool for analysing the influence of parameter values on the location 
and stability of equilibrium points is bifurcation analysis.
Bifurcation analysis is done routinely with numerical continuation methods 
for one adjustable bifurcation parameter \citep{Kuznetsov1995}.
Methods for numerical bifurcation analysis in several parameters are now being
developed \citep{Henderson2007,StiefsGro2007}, but due to practical considerations,
they remain limited to two or three adjustable bifurcation parameters.

The challenge to find parameter values for bifurcations is of particular relevance in the
area of biological systems.
The main reasons for this are that biological function is often based on complex
dynamical behaviour, and that parameters can vary within a large range due to
environmental or internal conditions.

There are many
examples where complex dynamical behaviour of a non-linear biological system can directly
be related to biological function.
Some examples from the specific area of biochemical signal transduction within
living cells are given by bistability in the mitogen activated protein kinase (MAPK) pathway 
to induce developmental processes \citep{FerrellXio2001}, 
rapid activation of caspases upon an over-threshold stimulus in programmed cell
death \citep{EissingCon2004}, and sustained oscillations in circadian clocks \citep{LeloupGol2003}.

Systems for biochemical signal transduction are usually modelled with non-linear
ordinary differential equations (ODEs).
Many models of biochemical systems contain a high number of model parameters,
usually even more parameters than state variables.
A major problem in understanding biochemical systems is that most of these
parameters are not very well known from measurements, and that they often vary 
significantly due to internal or environmental conditions of the cell.
Thus analysing the influence of uncertain or varying parameters on stability
properties is a fundamental issue towards understanding dynamical behaviour of biochemical systems.
Moreover, to avoid overlooking relevant effects it is necessary to 
consider simultaneous changes in all adjustable parameters \citep{StellingSau2004,KimPos2006}.

The requirement of looking at simultaneous changes in several parameters makes the application of
classical continuation methods problematic, as these require to define a line in parameter space along
which equilibrium points are tracked.
A good choice of this line is essential to obtain meaningful results, yet this choice
is often done by intuitive understanding of the system in the better case or iterative trials
in the worse.
Often only a single parameter is varied at a time, but then again
the choice of the parameter to vary is not trivial and needs to be done
for example via sensitivity considerations.

In this paper, we present a new method to locate points in a possibly high-dimensional parameter
space for a change in stability properties of equilibrium points, often hinting to either emergence 
or loss of complex dynamical behaviour. 
The method is based on considering the dynamical system as a closed loop feedback system. 
It is then possible to study properties of the original system in terms of 
an adequately defined open loop system.
If the open loop system is well chosen, then its dynamical behaviour is much simpler than that
of the closed loop system.
This simplification makes it possible to come to conclusions that could not
be obtained from the closed loop system alone.
In particular, we show how to classify parameter values where the closed loop system can
undergo local bifurcations of equilibrium points, based on an analysis of the open loop system.
The obtained conditions are used to develop a numerical method for searching
parameter values that lead to a change in stability properties.
In theory, this can be done for parameter spaces of arbitrary dimension, as neither the conditions nor the
algorithm we use depend on the dimension of the parameter space.
We consider only codimension one bifurcations, as they are the case that is generically
encountered in non-linear systems.

We make use of the fact that for stability considerations, it is 
sufficient to look at a linear approximation of the system close to the
equilibrium. The linearised system is transformed to the frequency domain for our
analysis.
The use of frequency domain methods for bifurcation analysis has already been introduced
by \citet[cited from \cite{MeesChu1979}]{Allwright1977}, and relevant results have also been
presented by Moiola and coworkers over the last decade \citep{MoiolaDes1991,MoiolaCol1997}.

Several authors have also studied the problem of finding bifurcations in systems
with many parameters using geometric tools.
Based on a description of vectors normal to a bifurcation manifold \citep{MonnigmannMar2002},
a method to search for locally closest bifurcations from a given reference
point was developed by \cite{Dobson2003}. 
These approaches can be seen as complementary to our results.
A recent application of the geometric concept to biological systems has been discussed by \cite{LuEng2006}.

%
%
%
%
%

Our paper is structured as follows. In Section~\ref{sec:methods}, we introduce the
loop breaking approach and provide the general tools which are necessary for our method.
The main results are presented in Section~\ref{sec:main-results}: a frequency domain theorem
on topological equivalence, an existence theorem for critical parameters and a
numerical algorithm to search for parameters yielding a change in dynamical behaviour.
Moreover, we shortly discuss the benefits of our approach compared to a straightforward
extension of classical tools.
As an application example, the method is used in Section~\ref{sec:application} to search for possible limit cycle
oscillations in an ODE model of a biochemical signal transduction system.

\section{The loop breaking concept}
\label{sec:methods}

\subsection{Problem setup}

Consider a parameter-dependent nonlinear differential equation given by
\begin{equation}
\label{eq:closed-loop-system}
\begin{aligned}
\dot x = F(x,p),
\end{aligned}
\end{equation}
with $x\in\Real^n$, $p\in\mathcal{P}\subset\Real^m$ and $F:\Real^n\times\mathcal{P}\rightarrow\Real^n$ a smooth vector field.

The system \eqref{eq:closed-loop-system} is studied locally at an equilibrium point. 
In what follows we frequently denote 
\begin{equation}
\begin{aligned}
\xi = (\bar x, p) \in \Real^n \times \mathcal P,
\end{aligned}
\end{equation}
where $\bar x$ is an equilibrium point and $p$ a corresponding parameter. 
We call $\xi$ an equilibrium--parameter pair of the system \eqref{eq:closed-loop-system} in the sense that
$\bar x$ is an equilibrium for the parameter $p$.
Let $\mathcal M$ be a
smooth connected $m$-dimensional manifold of equilibrium--parameter pairs
in $\Real^n\times\mathcal{P}$, i.e.
\begin{equation}
\label{eq:ep-manifold}
\begin{aligned}
\forall \xi \in \mathcal{M}: F(\bar x, p) =0.
\end{aligned}
\end{equation}
In the simplest case, there is a unique equilibrium point for each $p\in\mathcal{P}$, and one could use
a function $\bar x(p)$ to characterise the manifold of equilibrium--parameter pairs more easily. However, the
approach taken here is more general and also allows to consider
e.g. saddle-node bifurcations, where existence
of a unique equilibrium for each $p\in\mathcal{P}$ is not given.
For most applications, $\mathcal M$ can just be considered to be defined by the
equilibrium point equation
\begin{equation*}
\begin{aligned}
F(x,p) = 0.
\end{aligned}
\end{equation*}
In some cases it may however be beneficial to reduce $\mathcal M$ using
analytical tools before the analysis presented in this paper, in order
to satisfy technical assumptions or to improve the numerics.

\subsection{Loop breaking and closed loop eigenvalues}

Mathematically, the system \eqref{eq:closed-loop-system} is said to contain
a feedback loop if the influence graph of its Jacobian $\frac{\partial F}{\partial x}$
contains a nontrivial loop \citep{CinquinDem2002a}. 
Let us now assume that \eqref{eq:closed-loop-system} contains a feedback loop.
This assumption is not restrictive, because without a feedback loop, the analytical
expressions for the eigenvalues in terms of parameters and the state variables
can be taken directly from the diagonal of the (possibly permuted) Jacobian $\frac{\partial F}{\partial x}$.
In this case, it is usually easy to find parameter values for a change in stability
properties of the equilibrium points.

In the feedback loop approach, an input--output system which corresponds to the
original system is obtained by breaking the feedback loop.
As seen from the following definition, the original system
can be recovered by closing the feedback loop again.

\begin{definition}
\label{def:loop-breaking}
A \emph{loop breaking} for the system \eqref{eq:closed-loop-system} is a pair $(f,h)$, where
$f:\Real^n\times\Real\times\mathcal{P}\rightarrow\Real^n$ is a smooth vector field and $h:\Real^n\rightarrow\Real$
is a smooth function, such that
\begin{equation}
\label{eq:loop-breaking}
\begin{aligned}
F(x,p) = f(x,h(x),p).
\end{aligned}
\end{equation}
\end{definition}

The corresponding open loop system is then given by the equation 
\begin{equation}
\label{eq:open-loop-system}
\begin{aligned}
\dot x &= f(x,u,p) \\
y &= h(x),
\end{aligned}
\end{equation}
and the closed loop system \eqref{eq:closed-loop-system} is recovered by letting $u=y$.
Note that there is a direct relation between equilibrium points in the closed and the open loop
system: for an equilibrium--parameter pair $(\bar x,p)$ of the closed loop
system \eqref{eq:closed-loop-system},
setting the input $u=h(\bar x)$ in the open loop system \eqref{eq:open-loop-system} leads to $(\bar x, p)$ being
an equilibrium--parameter pair of \eqref{eq:open-loop-system}. We denote $\bar{u} = h(\bar{x})$.

To deal with the question whether different equilibrium-parameter pairs in $\mathcal{M}$ can have
different stability properties, it is reasonable to consider the linear approximation
of the system \eqref{eq:closed-loop-system} close to some pair $\xi\in\mathcal{M}$.
Only the pairs where the Jacobian $\frac{\partial F}{\partial x}(\xi)$ has eigenvalues
on the imaginary axis are candidate points for local bifurcations.
Any such pair $\xi$ is called a \emph{critical point}, and is denoted as $\xi_c$.

The linear approximation for the open loop system \eqref{eq:open-loop-system}
in the neighbourhood of the equilibrium--parameter pair $\xi\in\mathcal{M}$ is given by
\begin{equation}
\label{eq:open-loop-linear}
\begin{aligned}
\dot z &= A(\xi) z + B(\xi) \mu \\
\eta &= C(\xi) z,
\end{aligned}
\end{equation}
where $z = x - \bar x$, $\eta = y - \bar u$, $\mu = u - \bar u$, 
$A(\xi) = \frac{\partial f}{\partial x}(\bar x,\bar{u}, p)$,
$B(\xi) = \frac{\partial f}{\partial u}(\bar x,\bar{u}, p)$,
$C(\xi) = \frac{\partial h}{\partial x}(\bar x)$.

The linear approximation of the closed loop system \eqref{eq:closed-loop-system}
can then be easily characterised as follows.
\begin{proposition}
\label{prop:closed-loop-jacobian}
The linear approximation of the system \eqref{eq:closed-loop-system} close
to $\xi\in\mathcal{M}$ is given by
\begin{equation}
\label{eq:closed-loop-linear}
\begin{aligned}
\dot z = \left( A(\xi)+B(\xi)C(\xi)\right) z = A_{cl}(\xi) z.
\end{aligned}
\end{equation}
\end{proposition}
\begin{proof}
This follows directly from the loop breaking definition \eqref{eq:loop-breaking} and the chain rule.
\end{proof}

The linearised open loop system \eqref{eq:open-loop-linear}
can also be described using its transfer function, which is defined as
\begin{equation}
\label{eq:transfer-function}
\begin{aligned}
G(\xi,s) = C(\xi)\left(sI-A(\xi)\right)^{-1}B(\xi) = \frac{\det\matrix{sI-A(\xi)&-B(\xi)\\C(\xi)&0}}{\det(sI-A(\xi))}
\end{aligned}
\end{equation}
with the complex variable $s\in\mathbb{C}$.

The following lemma is a tool to characterise eigenvalues of the closed loop
system \eqref{eq:closed-loop-system} by analysing the open loop system \eqref{eq:open-loop-system}.
\begin{lemma}
\label{lem:eigenvalues}
$s_0\in\mathbb{C}$ is an eigenvalue of $A_{cl}(\xi)$, if and only if one of the
following conditions holds:
\renewcommand{\theenumi}{(\roman{enumi})}
\begin{enumerate}
\item $s_0$ is not an eigenvalue of $A(\xi)$ and $G(\xi,s_0)=1$;
\item $s_0$ is an eigenvalue of $A(\xi)$ and $\det\matrix{s_0 I-A(\xi)&-B(\xi)\\C(\xi)&0}=0$.
\end{enumerate}
\end{lemma}

The proof is provided in the appendix. In the following, Lemma~\ref{lem:eigenvalues} is
used with $s_0$ on the imaginary axis, to characterise critical points $\xi_c$
with the condition $G(\xi_c,s_0)=1$.

\subsection{Critical frequencies and imaginary closed loop eigenvalues}


In this section, the transfer function $G$ is represented
as a complex rational function with real coefficients, i.e.
\begin{equation}
\begin{aligned}
G(\xi,s) = \frac{k(\xi) q(\xi,s)}{r(\xi,s)},
\end{aligned}
\end{equation}
where $k(\xi)\in\Real$ and $q(\xi,s)$, $r(\xi,s)$ are polynomials in $s$ with 
real scalar functions of $\xi$ as coefficients.

Moreover, we make the following technical assumption.
\begin{enumerate}
\item[\textbf{(A1)}] 
The transfer function $G(\xi,\cdot)$ does not have
poles or zeros on the imaginary axis for any $\xi\in\mathcal{M}$, i.e.
\begin{equation}
\label{eq:open-loop-char-condition}
\begin{aligned}
\forall \xi \in \mathcal{M}\ \forall \omega \in\Real: 
k(\xi) q(\xi,j\omega)\neq 0 \textnormal{ and } r(\xi,j\omega)\neq 0.
\end{aligned}
\end{equation}
In addition, the degrees of $q(\xi,s)$ and $r(\xi,s)$ in $s$ are constant with respect to $\xi\in\mathcal{M}$.
\end{enumerate}
Starting from the premise that we are interested in stability changes produced by changing
the characteristics of the feedback loop that was broken in~\eqref{eq:loop-breaking},
this assumption is usually satisfied.

The notion of a critical frequency which is introduced in the next definition
will be useful to compute possible eigenvalues of the closed loop system
\eqref{eq:closed-loop-linear} on the imaginary axis.
\begin{definition}
\label{def:crit-freq}
$\omega_c\in\Real$ is said to be a \emph{critical frequency} for the transfer function $G(\xi,s)$,
if 
\begin{equation}
\label{eq:wcrit-def}
\begin{aligned}
G(\xi,j\omega_c) \in \Real.
\end{aligned}
\end{equation}

\end{definition}

Obviously, different values of $\xi$ will result in different critical frequencies.
For a specific $\xi$, all critical frequencies are given by the solutions of the equation
\begin{equation}
\label{eq:wcrit-computation}
\begin{aligned}
\operatorname{Im}(q(\xi,j\omega_c)r(\xi,-j\omega_c))=0,
\end{aligned}
\end{equation}
which is a scalar polynomial equation in $\omega_c$, with coefficients that are real scalar functions
of $\xi$.

We define the set of all critical frequencies for a specific $\xi$ as
\begin{equation}
\label{eq:critical-freq-set}
\begin{aligned}
\Omega_c(\xi) = \left\lbrace \omega\in\Real \mid \operatorname{Im}(q(\xi,j\omega)r(\xi,-j\omega))=0 \right\rbrace.
\end{aligned}
\end{equation}
Because only the imaginary part is considered, the polynomial in~\eqref{eq:wcrit-computation} is odd.
The following properties of the set $\Omega_c(\xi)$ can then be shown easily.
\begin{proposition}
For any $\xi\in\mathcal{M}$, the set $\Omega_c(\xi)$ satisfies the conditions
\renewcommand{\labelenumi}{(\roman{enumi})}
\begin{enumerate}
\item $0\in\Omega_c(\xi)$;
\item $\omega_c\in\Omega_c(\xi)$ implies that $-\omega_c\in\Omega_c(\xi)$;
\item either $\Omega_c(\xi)=\Real$ or $\Omega_c(\xi)$ has finitely many elements.
\end{enumerate}
\end{proposition}
Note that $\Omega_c(\xi)=\Real$ whenever $\operatorname{Im}(q(\xi,j\omega)r(\xi,-j\omega))$ is the zero polynomial,
which in turn is the case whenever only the even powers of $s$ in the polynomials $q(\xi,s)$ and $r(\xi,s)$
have nonzero coefficients. 
However, this typically contradicts assumption (A1), so we will
not consider this case specifically.

The relevance of critical frequencies for existence of eigenvalues on
the imaginary axis is shown by the following result.
\begin{proposition}
Assume that (A1) is satisfied. If $j\omega_c$ with $\omega_c\in\Real$
is an eigenvalue of $A_{cl}(\xi)$, then $\omega_c\in\Omega_c(\xi)$.
\end{proposition}
\begin{proof}
By assumption (A1), $j\omega_c$ is not an eigenvalue of $A(\xi)$.
By Lemma~\ref{lem:eigenvalues}, we have $G(\xi,j\omega_c)=1$ and thus $\omega_c\in\Omega_c(\xi)$.
\end{proof}

The concept of critical frequencies can be understood intuitively when
considering the Nyquist curve of the transfer function $G(\xi,j\omega)$. A critical frequency is
any value $\omega_c$ at which the Nyquist curve crosses the real axis. 
This is obviously a necessary condition for having $G(\xi,j\omega_c)=1$, which corresponds
to the existence of an eigenvalue on the imaginary axis as shown in Lemma~\ref{lem:eigenvalues}.
Our concept is thus closely related to the idea of the gain margin for robustness
analysis of linear control systems \citep{SkogestadPos1996}.

Since a variation of the equilibrium--parameter pair $\xi$
influences the polynomial equation~\eqref{eq:wcrit-computation}, the
set of critical frequencies $\Omega_c(\xi)$ may change significantly with $\xi$.
In particular,
the number of elements in $\Omega_c(\xi)$ in general needs not to be constant with
respect to $\xi$, which complicates the analysis.
However, one can show that there is a minimal number of critical frequencies of
the transfer function $G(\xi,s)$, which depends on the number of open loop poles and zeros and 
whether they are located in the right or the left half-plane.
To this end, define the number 
\begin{equation}
\label{eq:alpha-number}
\begin{aligned}
\alpha = \vert p_+ - p_- + z_- - z_+ \vert, 
\end{aligned}
\end{equation}
where $p_+$ ($p_-$) is the number
of poles of $G(\xi,\cdot)$ and $z_+$ ($z_-$) is the number
of zeros of $G(\xi,\cdot)$ in the right (left) half complex plane.
Under assumption (A1), $\alpha$ is constant with respect to $\xi\in\mathcal{M}$. 
The number of elements in the set of critical frequencies can now be characterised
by $\alpha$.

\begin{proposition}
\label{prop:minimality}
Let $\alpha$ be defined by \eqref{eq:alpha-number} and assume that (A1)
is satisfied.
Then, for any $\xi\in\mathcal{M}$,
$\Omega_c(\xi)$ has at least $\alpha$ distinct elements, if $\alpha$ is odd, and at least
$\alpha-1$ distinct elements, if $\alpha$ is even.
\end{proposition}

The proof is presented in the appendix.

The above result is used to formulate the property of minimality for the set of critical
frequencies.
\begin{definition}
\label{def:minimality}
Under the assumptions of Proposition~\ref{prop:minimality},
the set of critical frequencies $\Omega_c(\xi)$ is called \emph{minimal}, if it contains
exactly $\tilde \alpha$ elements, where
\begin{equation}
\begin{aligned}
\tilde\alpha = \left\lbrace 
\begin{aligned}
\alpha,&\quad\textnormal{if $\alpha$ is odd} \\
\alpha-1,&\quad\textnormal{if $\alpha$ is even.}
\end{aligned}\right.
\end{aligned}
\end{equation}
\end{definition}

This definition is applied in the second technical assumption we are going to
make use of.
\begin{enumerate}
\item[\textbf{(A2)}] 
The set of critical frequencies $\Omega_c(\xi)$ is minimal for any $\xi\in\mathcal{M}$.
\end{enumerate}

If (A2) holds, we can label the roots of the polynomial 
equation~\eqref{eq:wcrit-computation} in a consistent way, writing
\begin{equation}
\label{eq:minimal-set-crit-freq}
\begin{aligned}
\Omega_c(\xi) = \left\lbrace \omega_c^1(\xi), \omega_c^2(\xi), \ldots ,
\omega_c^{\tilde\alpha}(\xi) \right\rbrace,
\end{aligned}
\end{equation}
where the $\omega_c^i$ are continuous functions of the equilibrium--parameter pair $\xi$ and
can be identified with different solution branches of the polynomial equation \eqref{eq:wcrit-computation}.

Given a transfer function $G(\xi,\cdot)$ and a corresponding set of
critical frequencies $\Omega_c(\xi)$, Proposition~\ref{prop:minimality} can be used to easily
check the minimality of $\Omega_c(\xi)$.
Graphically, a sufficient condition for minimality of $\Omega_c(\xi)$ is that the Nyquist
curve $G(\xi,j\omega)$ encircles the origin monotonically as $\omega$ varies from $-\infty$
to $\infty$.

\section{Main results}
\label{sec:main-results}

\subsection{Topological equivalence of equilibria}

Changes in stability properties of equilibrium points are most easily studied
using the concept of topological equivalence.
Here, we use a definition for hyperbolic equilibrium points only
(see \cite{Kuznetsov1995} for more details).

\begin{definition}
\label{def:topological-equivalence}
Let $\xi_1,\xi_2\in\mathcal{M}$ be two hyperbolic equilibrium--parameter pairs
of the system \eqref{eq:closed-loop-system}.
$\xi_1$ and $\xi_2$ are said to be \emph{topologically equivalent}, if the Jacobians $\frac{\partial F}{\partial x}(\xi_1)$ 
and $\frac{\partial F}{\partial x}(\xi_2)$ have the same number of eigenvalues in the left
and right half-plane.
\end{definition}

It is a well known result from dynamical systems theory that the topological equivalence of all
equilibria in two systems is a necessary condition for topological equivalence of the flows.
Let us consider two variants of the system \eqref{eq:closed-loop-system}, one with parameter values
$p_1$ and the other with parameter values $p_2$.
In the simple case when there is only one equilibrium point in each variant of the system,
corresponding to the pairs $\xi_1$ and $\xi_2$, topological equivalence of $\xi_1$ and $\xi_2$
is a necessary condition for topological equivalence of the flows.
In applications, we are often interested in finding parameter values $p_2$ such that the system
\eqref{eq:closed-loop-system} changes its dynamical behaviour when varying parameters from initial
values $p_1$ to $p_2$.
For this problem, it is sufficient to find pairs $\xi_1$ and $\xi_2$ which are not
topologically equivalent.
Due to the continuous dependence of eigenvalues on parameters,
this can only happen when the Jacobian $\frac{\partial F}{\partial x}(\xi_c)$ has eigenvalues
on the imaginary axis for some critical point $\xi_c \in \mathcal{M}$.
At this point, we can make use of the methodology developed in the previous section.

To this end, consider the set of critical frequencies $\Omega_c(\xi)$ for a specific
value of the equilibrium--parameter pair $\xi$.
Define the number $\beta(\xi)$ to be the number of elements $\omega_c$ in $\Omega_c(\xi)$ such
that $G(\xi,j\omega_c)>1$, i.e.
\begin{equation}
\begin{aligned}
\beta(\xi) = \operatorname{card}\left\lbrace \omega_c\in\Omega_c(\xi) \mid G(\xi,j\omega_c)>1 \right\rbrace,
\end{aligned}
\end{equation}
where $\operatorname{card}\mathcal{S}$ denotes the number of elements in the set $\mathcal{S}$.

Geometrically, if $\Omega_c(\xi)$ is minimal, $\beta(\xi)$ gives the winding number 
of the graph of $G(\xi,j\omega)$ around the
point $1$ in the complex plane (see Lemma~\ref{lem:winding-number} in the Appendix).
The argument principle can then be used to characterise topologically equivalent equilibrium--parameter
pairs of the system \eqref{eq:closed-loop-system} via the number $\beta(\xi)$.
We first give some intermediate results as Lemmas before presenting the main theorem.
The Lemmas are proven in the appendix.
\begin{lemma}
\label{lem:ordered-crit-freq}
If the set of critical frequencies $\Omega_c(\xi)$ is minimal, then in the ordered sequence
of critical frequencies 
$\omega_c^1(\xi) < \omega_c^2(\xi) < \cdots < \omega_c^\alpha(\xi)$,
we have \[G(\xi,j\omega_c^i(\xi)) G(\xi,j\omega_c^{i-1}(\xi)) < 0\] where $i=2,\ldots,\alpha$.
\end{lemma}
\begin{lemma}
\label{lem:winding-number}
Under the assumptions of Theorem~\ref{theo:topological-equivalence}, the
winding number of the image of the Nyquist curve $\Gamma$ under the transfer
function $G(\xi_i,\cdot)$, $i=1,2$, around the point $1$ is given by
\begin{equation*}
\begin{aligned}
\vert wn(G(\xi_i,\Gamma),1) \vert = \beta(\xi_i).
\end{aligned}
\end{equation*}
\end{lemma}
\begin{lemma}
\label{lem:same-winding-number}
Under the assumptions of Theorem~\ref{theo:topological-equivalence}, we have
\begin{equation*}
\begin{aligned}
\big\vert wn(G(\xi_1,\Gamma),1) - wn(G(\xi_2,\Gamma),1) \big\vert = 
\big\vert \beta(\xi_1) - \beta(\xi_2) \big\vert.
\end{aligned}
\end{equation*}
\end{lemma}
\begin{theorem}
\label{theo:topological-equivalence}
Assume that (A1) is satisfied. Let $\xi_1,\xi_2\in\mathcal{M}$ be
two hyperbolic equilibrium--parameter pairs of \eqref{eq:closed-loop-system}
such that $\Omega_c(\xi_1)$ and $\Omega_c(\xi_2)$ are minimal. Then $\xi_1$ and $\xi_2$ are
topologically equivalent, if and only if
\begin{equation*}
\begin{aligned}
\beta(\xi_1) = \beta(\xi_2).
\end{aligned}
\end{equation*}
\end{theorem}
The proof is given in the appendix.

\subsection{Existence of marginally stable equilibria}

Let us now turn to the problem of how to find parameter values for which a change 
in stability properties of equilibria can happen.
This is equivalent to searching for critical points $\xi_c$ at which the Jacobian
$\frac{\partial F}{\partial x}(\xi_c)$ has an eigenvalue on the imaginary axis.
This typically means
that $\xi_c$ is part of a submanifold of $\mathcal{M}$ that separates regions
of topological equivalence, 
and in view of Theorem~\ref{theo:topological-equivalence} there are
typically equilibrium--parameter pairs $\xi_1$ and $\xi_2$ close to $\xi_c$ such
that $\beta(\xi_1) \neq \beta(\xi_2)$.
Equivalently, for a specific critical frequency $\omega_c^i$, the transfer
function value $G(\xi,j\omega_c^i(\xi))$ has to cross the value $1$ when
$\xi$ varies continuously along a path from $\xi_1$ to $\xi_2$.
These observations are formalised in the following theorem.

\begin{theorem}
\label{theo:crit-param-existence}
Assume that (A1) and (A2) are satisfied.
There exists a critical point $\xi_c\in\mathcal{M}$ such that 
$\frac{\partial F}{\partial x}(\xi_c)$ has an eigenvalue on the imaginary axis,
if and only if there exist $\xi_1, \xi_2 \in\mathcal{M}$ such that,
for some $i\in\lbrace 1,2,\ldots,\tilde\alpha\rbrace$,
\begin{equation}
\label{eq:crit-param-existence}
\begin{aligned}
G(\xi_1,j\omega_c^i(\xi_1)) \leq 1 \leq G(\xi_2,j\omega_c^i(\xi_2)),
\end{aligned}
\end{equation}
where $\omega_c^i(\xi) \in \Omega_c(\xi)$.
In that case, $\pm j\omega_c^i(\xi_c)$ is an eigenvalue of $\frac{\partial F}{\partial x}(\xi_c)$.
\end{theorem}
\begin{proof}
By Lemma~\ref{lem:eigenvalues}, our assumptions assure that a point $\xi_c$ is critical if and
only if $G(\xi_c,j\omega_c^i(\xi_c))=1$.

\emph{Necessity.} Under the condition $G(\xi_c,j\omega_c)=1$, take $\xi_1=\xi_2=\xi_c$ and
\eqref{eq:crit-param-existence} follows trivially.

\emph{Sufficiency.} Let $\xi_1$ and $\xi_2$ be such that \eqref{eq:crit-param-existence}
holds. Connectivity of $\mathcal M$ implies that there is a path from $\xi_1$ to $\xi_2$
in $\mathcal M$. Continuity of the critical frequency $\omega_c^i(\xi)$ and the transfer
function coefficients result in continuity of $G(\xi,j\omega_c^i(\xi))$ with respect to
$\xi$. 
This implies existence of $\xi_c$ such that $G(\xi_c,j\omega_c^i(\xi_c))=1$ 
along any path from $\xi_1$ to $\xi_2$.
\end{proof}

The proof shows that the critical point $\xi_c$ is usually far from unique.
It may be unique if $\mathcal{M}$ is of dimension one, i.e. there is only one free
parameter to vary.
In general, one will expect that there is a submanifold of critical points in
$\mathcal{M}$ separating regions which represent different topological equivalence classes,
where $\xi_1$ is an element of one such class, and $\xi_2$ is an element of the other class.
On this submanifold, the bifurcations that can be encountered generically are
codimension one bifurcations.
Therefore the bifurcation condition provided by Theorem~\ref{theo:crit-param-existence}
is mainly useful in the search for codimension one bifurcations, although
condition \eqref{eq:crit-param-existence} also holds for
bifurcations of higher codimension. 

Also note that the assumptions (A1) and (A2) are sufficient, but not necessary
in Theorem~\ref{theo:crit-param-existence}.
For example, (A2) may be violated, but if the additional critical frequencies
do not lead to a change in the winding number of the transfer function graph
around the point $1$, then the conclusion is still valid.

Using Theorem~\ref{theo:crit-param-existence}, it is easily possible to distinguish
dynamical bifurcations from static bifurcations.
In fact, if $i$ is chosen such that the critical frequency $\omega_c^i(\xi) = 0$ 
is considered, $A_{cl}(\xi_c)$ has a
zero eigenvalue, which generically corresponds to a saddle-node bifurcation.
If a critical frequency $\omega_c^i(\xi) \neq 0$ is considered, $A_{cl}(\xi_c)$ has
conjugated imaginary eigenvalues, and one will generically get a Hopf bifurcation.

A graphical illustration of Theorem~\ref{theo:crit-param-existence} is given
in Fig.~\ref{fig:nyquist}.
The relation to the classical Nyquist stability criterion also becomes
clear from this figure.

\begin{figure}
\begin{center}
\includegraphics[width=0.6\textwidth]{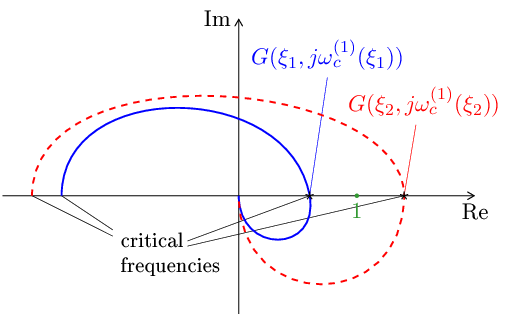}
\end{center}
\caption{Illustration of Theorem~\ref{theo:crit-param-existence} in the Nyquist plot.
Full line: $G(\xi_1,j\omega)$, dashed line: $G(\xi_2,j\omega)$, both for $\omega \geq 0$.
The theorem asserts existence of a Hopf bifurcation on any path between $\xi_1$
and $\xi_2$.}
\label{fig:nyquist}
\end{figure} 

\subsection{An algorithm for a numerical parameter search}
\label{ssec:numerical-parameter-search}

In this section, we discuss an algorithm to search for parameter values
that will lead to a change in stability properties of an equilibrium point.
We assume that a starting parameter $p_1$ and a corresponding equilibrium
$\bar x_1$ are known, which we combine in the pair $\xi_1 = (\bar x_1, p_1)
\in \mathcal M$.
It is reasonable to assume that the pair $\xi_1$ is not critical, otherwise
it is usually straightforward to find parameter values yielding equilibrium points
with different stability properties.
Moreover, the manifold $\mathcal M$ is assumed to be defined by a nonlinear
equation of the form $\varphi(\xi) = 0$.
Often, one can directly use $\varphi=F$, but sometimes a
modification is useful to exclude some solutions if the equation $F(x,p)=0$
is known to have multiple solutions.
The aim of the algorithm is to find an equilibrium--parameter pair $\xi_2 \in \mathcal{M}$
such that $\xi_2$ is not topologically equivalent to $\xi_1$.
The main theoretical basis of the algorithm is the result of 
Theorem~\ref{theo:crit-param-existence}.
Thus it is also possible to search specifically for either static or
dynamic bifurcations on a path from $\xi_1$ to $\xi_2$ by choosing
an appropriate critical frequency.

In order to put the problem in the framework developed in this paper, a
loopbreaking for the system~\eqref{eq:closed-loop-system} has to
be defined.
Then, by looking at the resulting transfer function $G(\xi_1,s)$,
possible changes in stability properties can be determined.
In particular one has to decide whether to search for a static or for a
dynamic bifurcation.
This leads to the choice of a critical frequency $\omega_c^i$ which
is to be considered in the algorithm.

Denote the transfer function value for the critical frequency
at a point $\xi$ as $\gamma(\xi)$.
At the starting point $\xi_1$, this value can be
computed as 
\begin{equation}
\begin{aligned}
\gamma(\xi_1) = G(\xi_1,j\omega_c^i(\xi_1)).
\end{aligned}
\end{equation}
Note that an analytical expression of the function $\gamma$ can be
derived directly, maybe with the support of computer algebra for more complex systems.
This derivation requires only basic algebraic manipulations, differentiation and
matrix inversion, which can all be done symbolically for typical system classes.
In particular, it is not required to have an analytical solution of the
equation $F(x,p) = 0$ to construct $\gamma$.

Now two cases have to be distinguished.
\begin{enumerate}
\item If $\gamma(\xi_1) < 1$, the algorithm searches a pair $\xi_2 \in \mathcal M$ such
that $\gamma(\xi_2) > 1$.
\item If $\gamma(\xi_1) > 1$, the algorithm searches $\xi_2 \in \mathcal M$ such
that $\gamma(\xi_2) < 1$.
\end{enumerate}

The algorithm we are using is best described by the term
\emph{gradient-directed continuation method}. 
Continuation methods \citep{RichterDeC1983}
are popular in numerical bifurcation analysis, where they are used to
trace the equilibrium curve in the combined state-parameter space.
In our algorithm, continuation is used to stay on the manifold
$\mathcal{M}$.
However, a continuation method alone is not sufficient, as
$\mathcal{M}$ is $m$-dimensional with typically $m>1$.
Thus, the continuation is complemented with a gradient ascent or descent
approach to achieve the desired value for $\gamma(\xi_2)$.

Since the algorithm is based on Theorem~\ref{theo:crit-param-existence},
assumptions (A1) and (A2) need to be checked.
Depending on the system under consideration, this may be a difficult problem
globally over the equilibrium-parameter manifold $\mathcal M$.
However, for the validity of the algorithm's results it is sufficient
that (A1) and (A2) are satisfied locally along the path used for the continuation.
These checks can be directly included in the algorithm.
If the assumptions are violated at one point, the algorithm issues a
warning message.
The results may still be valid, because (A1) and (A2) are only
sufficient, but not necessary conditions for Theorem~\ref{theo:crit-param-existence}.
However, the results need to be checked separately in this case.

In detail, the algorithm works as follows. We are discussing
case 1 only, small extensions are required for dealing with
both cases.
\begin{enumerate}
\item \textbf{Initialisation.} Set $\xi^{(0)} = \xi_1$.
Choose numerical parameters: $\Delta\gamma$ for the minimal required change in $\gamma(\xi)$ per
iteration, $\delta^{(0)}$ as the initial step size and $\delta_{min}$ ($\delta_{max}$)
as minimal (maximal) step size.
\item \textbf{Checking assumptions.}
(A1) is checked locally by computing the poles and zeros of $G(\xi^{(i)},s)$.
(A2) is checked locally by computing the critical frequencies $\Omega_c(\xi^{(i)})$
and applying Proposition~\ref{prop:minimality}.
If the assumptions are not satisfied, output a warning.
\item \textbf{Prediction step.} This step assures the desired
increase in $\gamma(\xi)$.
\begin{enumerate}
\item Compute the gradient $\nabla\gamma\left(\xi^{(i)}\right)$.
\item Compute the subspace which is tangent to $\mathcal{M}$ in
the point $\xi^{(i)}$:
\begin{equation}
\begin{aligned}
T_{\xi^{(i)}}\mathcal{M} = \operatorname{null}\frac{\partial\varphi}{\partial \xi}\left(\xi^{(i)}\right).
\end{aligned}
\end{equation}
\item Project $\nabla\gamma\left(\xi^{(i)}\right)$ on $T_{\xi^{(i)}}\mathcal{M}$:
\begin{equation}
\label{eq:gradient-projection}
\begin{aligned}
v^{(i)} = \operatorname{Proj}\left(\nabla\gamma\left(\xi^{(i)}\right),T_{\xi^{(i)}}\mathcal{M}\right).
\end{aligned}
\end{equation}
\item Set the predicted point 
\begin{equation}
\begin{aligned}
\xi^{(i+1)}_{pre} = \xi^{(i)} + \delta^{(i)} v^{(i)}.
\end{aligned}
\end{equation}
Step size control is used in the sense that $\delta^{(i)}$ is varied to assure
that the condition
\begin{equation}
\label{eq:gamma-increase}
\begin{aligned}
\gamma\left(\xi^{(i+1)}_{pre}\right) - \gamma\left(\xi^{(i)}\right) \geq \Delta\gamma
\end{aligned}
\end{equation}
is satisfied, while keeping $\delta_{min} \leq \delta^{(i)} \leq \delta_{max}$.
\end{enumerate}
\item \textbf{Correction step.} Generally, $\xi^{(i+1)}_{pre} \notin \mathcal{M}$, so
a correction step is required to achieve $\xi^{(i+1)} \in \mathcal{M}$.
To this end, the Gauss-Newton method is used to solve the nonlinear equation
\begin{equation}
\begin{aligned}
\varphi(\xi^{(i+1)}) &= 0 \\
\gamma\left(\xi^{(i+1)}\right)  &= \gamma\left(\xi^{(i+1)}_{pre}\right)
\end{aligned}
\end{equation}
for $\xi^{(i+1)}$, where $\xi^{(i+1)}_{pre}$ is used as starting point for the Gauss-Newton algorithm.
If the Gauss-Newton algorithm converges, the algorithm takes the solution as value for $\xi^{(i+1)}$ 
and proceeds to the next step.
Otherwise, the algorithm reduces the step size $\delta^{(i)}$ and goes back to 2d).
\item \textbf{Finishing criterion.} Compute $\gamma\left(\xi^{(i+1)}\right)$. If
$\gamma\left(\xi^{(i+1)}\right)> 1$, finish successfully, otherwise iterate
to step 2.
\end{enumerate}

If the algorithm finishes successfully, it does so in a finite
number of steps with a previously known upper bound 
due to step size control via inequality~\eqref{eq:gamma-increase}.

However, in the same way as classical continuation methods, the algorithm may
fail if the Gauss-Newton algorithm in step 3 does not converge, and the
step size $\delta^{(i)}$ may not be reduced further due to the constraint
$\delta_{min} \leq \delta^{(i)}$ at the same time.
This problem may appear if the system is numerically ill-conditioned, 
but can typically be avoided by choosing a smaller value for either $\Delta\gamma$
or for $\delta_{min}$, with the drawback of increased computational effort.
Also, it can in general not be excluded that the function $\gamma(\xi)$ 
has local extrema, which may pose problems to the algorithm.
Such problem may be detected numerically from the vector $v^{(i)}$ taking very
small values.
However, in several applications we have not encountered this problem so far.

The algorithm as described above does not consider constraints on the parameters $p$. 
Such constraints can be included by slight modifications in steps 2 and 3.
If the border of the set $\mathcal{P}$ is approached during the iteration, the modified
algorithm projects the gradient $\nabla\gamma\left(\xi^{(i)}\right)$ on the intersection
of the tangent to $\mathcal{M}$ and the tangent to the border of $\mathcal{P}$.
With additional step size control, a constraint violation is then avoided.

\subsection{Discussion of the feedback loop approach}

Two key steps in the approach we have taken are the transformation of the
problem to the frequency domain and the consideration of the critical frequencies.
These steps require an elaborate setup and therefore need to be well justified.

Approaching the given problem in the time domain would typically require to
deal with the eigenvalues of the Jacobian $\frac{\partial F}{\partial x}$ on the
considered manifold of equilibrium points.
In particular, it would require to consider how the eigenvalues change
if the parameters change.
Going to the frequency domain will typically reduce the number of variables
that are to be tracked with changes, because there are typically less
critical frequencies than eigenvalues.
For a minimum phase system, the number of critical frequencies is not more than the
relative degree of the transfer function $G(\xi,s)$.
Moreover, the position of eigenvalues has to be tracked in the
two dimensions of the complex plane, whereas the transfer function values at critical
frequencies are always real numbers.
In particular, it would be difficult to estimate from the eigenvalues of the system for some starting
parameters, which pair of eigenvalues should be pushed to the imaginary axis in order
to obtain a Hopf bifurcation.
Using the frequency domain approach, the critical frequency for which the transfer function
value should be pushed towards 1 can typically be determined easily.


In classical bifurcation analysis, so called bifurcation test functions are used
to check whether a bifurcation may occur when going from one parameter 
value to another one \citep{Kuznetsov1995}.
The test function $\Psi$ is defined such that $\Psi(\xi_c) = 0$ if the bifurcation
that is tested for occurs at $\xi_c$.
Bifurcations are detected by the test function $\Psi(\xi)$ changing
sign when going from one point to the other, i.e.\ if $\Psi(\xi_1)\Psi(\xi_2) < 0$,
then a bifurcation occurs between $\xi_1$ and $\xi_2$.
For bifurcations of codimension one, suitable test functions are known and are
routinely used in numerical continuation algorithms.
Note that in the frequency domain approach, the expression $G(\xi,j\omega_c(\xi))-1$
is a test function for a generic saddle-node bifurcation, if we consider $\omega_c = 0$,
and it is a test function for a generic Hopf bifurcation when considering $\omega_c \neq 0$.
Computing classical test functions for a given point $\xi$ requires a similar or slightly less
computational effort as computing the transfer function values at the critical frequency.
So we need to justify why we do not use classical test functions for bifurcation
search in a high-dimensional parameter space.

Because classical continuation methods cannot be used in a high-dimensional parameter space,
one has to look for different approaches.
A naive approach to find parameters for a bifurcation would be to solve directly
the equations
\begin{equation}
\label{eq:naive-search}
\begin{aligned}
\Psi(\xi) &= 0 \\
F(\xi) &= 0.
\end{aligned}
\end{equation}
However, in most cases this will be numerically infeasible with classical test
functions, even if the combined parameter state space is of very low dimension.
A more sophisticated approach could use basically the same algorithm that we have
presented in Section~\ref{ssec:numerical-parameter-search}, the gradient-directed
continuation method, and just use the gradient of a classical bifurcation test function instead
of the gradient of the transfer function $G(\xi,j\omega_c(\xi))$.
We have also implemented this approach for several examples, but run into numerical problems
for any system of medium complexity.
In particular, the analysis presented in the next section did not work with a classical
bifurcation test function for a Hopf bifurcation due to numerical problems.
These problems seems to be related to our observation that the value of the classical bifurcation test
function seems to be numerically much less well behaved with respect to parameter variations
than the transfer function value at critical frequencies.


\section{Application to a biological signalling system}
\label{sec:application}

In this section, we apply the theoretical results and the numerical algorithm described
in the previous section to a system for biochemical signal transduction. 
A central element of the signal transduction in eukaryotic cells is the mitogen 
activated protein kinase (MAPK) cascade. 
It appears in several signalling pathways and is related to cell differentiation,
proliferation, and response to external stress.
Several ODE models for this system have been proposed during the last decade \citep{OrtonStu2005}.

The MAPK cascade consists of three layers of kinase proteins, where each kinase activates
the next layer, and the last layer corresponds to the output of the cascade.
We consider the MAPK cascade as it appears in the EGF (epidermal growth factor) receptor pathway
\citep{BrightmanFel2000}. There, the three kinases in the order how activation
proceeds are Raf, MEK (MAPK/ERK kinase), and ERK (extracellular-signal-regulated kinase).
Active ERK phosphorylates and inhibits SOS (son of sevenless homologue), which is required in the activation
of Raf \citep{BrightmanFel2000}. This constitutes a negative feedback loop
in the system.
The model we use here is a slight simplification of a model suggested by \cite{Kholodenko2000}, and
it is also a subsystem of the EGF pathway as modelled by \cite{BrightmanFel2000}.
A cartoon of the biochemical reactions incorporated in the model is shown in Fig.~\ref{fig:mapk-model}.

\begin{figure}[tb]
\begin{center}
\includegraphics[width=12cm]{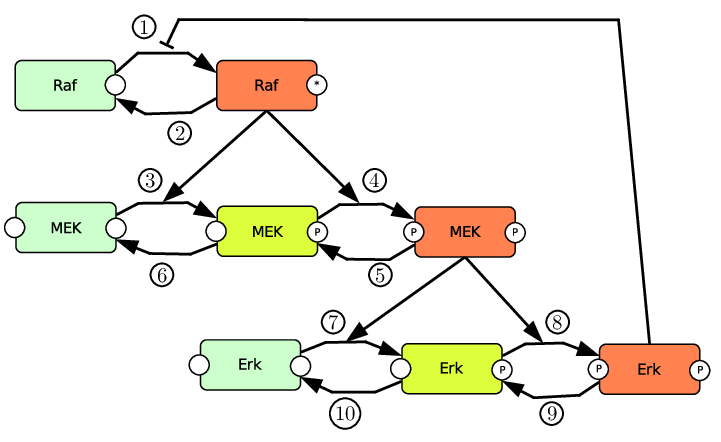}
\end{center}
\vspace{0.3cm}
\caption{Illustration of the MAPK cascade model with reaction numbers.}
\label{fig:mapk-model}
\end{figure} 

In the equations, the concentrations of phosphorylated kinases are denoted as 
$x_{11}=$ [Raf*], $x_{21}=$ [MEK-P], $x_{22}=$ [MEK-PP],
$x_{31}=$ [ERK-P] and $x_{32}=$ [ERK-PP]. The concentrations of unphosphorylated inactive kinases
Raf, MEK and ERK need not be included as state variables, as they can be computed via the conservation
laws
\begin{equation*}
\begin{aligned}
\textnormal{[Raf]} + x_{11} &= x_{1t} \\
\textnormal{[MEK]} + x_{21} + x_{22} &= x_{2t} \\
\textnormal{[ERK]} + x_{31} + x_{32} &= x_{3t},
\end{aligned}
\end{equation*}
where $x_{1t}$, $x_{2t}$ and $x_{3t}$ are parameters for the total concentrations of kinases,
which are constant.
Table~\ref{tab:mapk-reaction-rates} shows the mathematical expressions for the reaction rates,
with numbers corresponding to the labels in Fig.~\ref{fig:mapk-model}.
Nominal parameter values for the simplified model have been adopted from \citep{Kholodenko2000},
and are shown in Table~\ref{tab:mapk-parameters} as $p_1$.

\begin{table}
\caption{Reaction rates in the MAPK cascade model}
\begin{tabular}{|c|c|}\hline
Reaction & Rate \\ \hline
$v_1$ & $V_1\frac{x_{1t}-x_{11}}{(1+x_{32}/K_{i})(K_{m1}+x_{1t}-x_{11})}$ \\[0.1cm] \hline
$v_2$ & $V_2\frac{x_{11}}{K_{m2}+x_{11}}$ \\[0.1cm] \hline
$v_3$ & $k_3 x_{11}(x_{2t}-x_{21}-x_{22})$ \\[0.1cm] \hline
$v_4$ & $k_4 x_{11} x_{21}$ \\[0.1cm] \hline
$v_5$ & $V_5 \frac{x_{22}}{K_{m5}+x_{22}}$ \\[0.1cm] \hline
$v_6$ & $V_6 \frac{x_{21}}{K_{m6}+x_{21}}$ \\[0.1cm] \hline
$v_7$ & $k_7 x_{22} (x_{3t}-x_{31}-x_{32})$ \\[0.1cm] \hline
$v_8$ & $k_8 x_{22} x_{31}$ \\[0.1cm] \hline
$v_9$ & $V_9 \frac{x_{32}}{K_{m9}+x_{32}}$ \\[0.1cm] \hline
$v_{10}$ & $V_{10}\frac{x_{31}}{K_{m10}+x_{31}}$ \\[0.1cm] \hline
\end{tabular} 
\label{tab:mapk-reaction-rates}
\end{table}

Using the reaction rates from Table~\ref{tab:mapk-reaction-rates}, the model
can be written as a system of five ODEs with 20 parameters:
\begin{equation}
\label{eq:mapk-model}
\begin{aligned}
\dot x_{11} &= v_1 - v_2 \\
\dot x_{21} &= v_3 + v_5 - v_4 - v_6 \\
\dot x_{22} &= v_4 - v_5 \\
\dot x_{31} &= v_7 +  v_9 - v_8 - v_{10} \\
\dot x_{32} &= v_8 - v_9.
\end{aligned}
\end{equation}

The only difference to Kholodenko's model is in the phosphorylation reactions
3, 4, 7 and 8. 
The original model uses Michaelis-Menten kinetics for these reactions,
whereas our simplified model uses mass action kinetics.
It can be argued that with the concentrations of all kinases being on the same order of
magnitude, the assumptions for using Michaelis-Menten kinetics in reactions 3, 4, 7 and
8 are not valid anyway, and one could aim to achieve a similar dynamical behaviour with the
simpler model structure where mass action kinetics are used.

Kholodenko has shown in simulations that the system can show limit cycle oscillations
for some parameter values. Due to the simplifications in four reaction rates, the model~\eqref{eq:mapk-model}
does not oscillate for nominal parameter values $p_1$.
Instead, the model has a stable equilibrium $\bar x_1$ for these parameter values.
Solutions of the model converge to the steady state within $20$ seconds, as depicted in Fig.~\ref{fig:mapk-solutions-p0}.

The question we deal with is whether parameters can be changed such that also
the simplified model shows sustained oscillations.
To answer this question, we apply the algorithm to search for
critical parameter values which is described in the previous section.

The first step in our analysis is to choose a suitable loop breaking. For the given system, an intuitive
approach is to break the loop at the feedback inhibition of reaction $v_1$ by ERK-PP.
Thus we choose $h(x) = x_{32}$ and replace $x_{32}$
by the input $u$ in the reaction rate $v_1$, thus obtaining the dynamics of the open loop system $f(x,u,p)$.

A linearisation of the open loop system
around the equilibrium point and a Laplace transformation give the transfer function 
$G(\xi,s)$, whose graph is shown in Fig.~\ref{fig:mapk-nyquist}.
The problem is now to find parameters $p_2$ with a corresponding unstable equilibrium point $\bar x_2$.
This can be done using the numerical algorithm presented in Section~\ref{ssec:numerical-parameter-search}.

The set of critical frequencies is minimal with $\alpha=3$, which can be seen from Fig.~\ref{fig:mapk-nyquist} by
the observation that the graph of $G(\xi_1,j\omega)$ encircles the origin monotonically and crosses
the real axis three times. The origin of the complex plane is not counted as 
crossing, as we have $\omega=\pm\infty$ there.
The only positive critical frequency is $\omega_c^3(\xi_1) = 0.017\,s^{-1}$, and we will
consider only $\omega_c^3$ in the search for destabilising parameters, because our
goal is to find a Hopf bifurcation.
The corresponding transfer function value
is $G(\xi_1,j\omega_c^3(\xi_1)) = \gamma(\xi_1) = 0.12$, corresponding to the equilibrium $\bar x_1$ being stable
in the closed loop system.

The goal for the parameter search algorithm is to find parameters such that
$\gamma(\xi_2) > 1$.
Then the corresponding equilibrium point $\bar{x}_2$ will not be topologically equivalent
to the nominal equilibrium $\bar x_1$ and we can expect a Hopf bifurcation when
varying parameters from the nominal value $p_1$ to the new value $p_2$.
To ensure that the algorithm does not stop at the bifurcation point, but continues to vary
parameters until the oscillations have reached a considerable amplitude, we try to achieve
$\gamma(\xi_2) \geq 1.5$ in the implementation used here.


In the application of the algorithm to this problem, we set the minimal change
in the transfer function value per iteration $\Delta\gamma=10^{-4}$ and
the initial step size $\delta^{(0)} = 10^{-4}$.
The Gauss-Newton algorithm in the correction step was constrained to 20 iterations,
but in the step size control the step size $\delta^{(i)}$ was already decreased 
if the Gauss-Newton algorithm required five or more iterations for convergence.
With these settings the algorithm finishes successfully after 276 iterations,
yielding the parameters $p_2$ and an equilibrium $\bar x_2$ with the transfer function value
$G(\xi_2,j\omega_c^3(\xi_2)) = 1.52$ and the critical frequency $\omega_c^3(\xi_2) = 0.0068\,s^{-1}$,
where $\xi_2 = (\bar x_2, p_2)$.
The parameter values in $p_2$ are shown in Table~\ref{tab:mapk-parameters}.
Note that although in principle all parameters could have been changed when going
from $p_1$ to $p_2$, the algorithm varies only 9 out of the 20 parameters by an
amount of more than 20 \%. Since the algorithm uses the gradient of the
transfer function value at the critical frequency, we can presume that the parameters
that have been varied by a larger amount
have higher influence on existence of oscillations than the other parameters.

\begin{table}[htb]
\caption{Reference parameters $p_1$ and parameters for instability $p_2$ in the
MAPK cascade model.}
\begin{tabular}{|c|c|c|c|l|} \hline
Param. & $p_1$ & $p_2$ & Unit & rel. change\\ \hline
   $V_1$ & 2.5 & 2.5 & nM/s & $1.00$\\ \hline
   $K_i$ & 9 & 18.9 & nM & $2.09$\\ \hline
  $K_{m1}$ & 10 & 8.1 & nM & $1.23^{-1}$\\ \hline
   $V_{2}$ & 0.25 & 0.17 & nM/s & $1.43^{-1}$\\ \hline
  $K_{m2}$ & 8 & 0.54 & nM & $14.8^{-1}$\\ \hline
   $k_{3}$ & 0.001 & $4.3\e{-4}$ & 1/(s nM) & $2.34^{-1}$\\ \hline
   $k_{4}$ &  0.001 & $5.7\e{-4}$ & 1/(s nM) & $1.76^{-1}$\\ \hline
   $V_{5}$ & 0.75 & 0.74 & nM/s & $1.01^{-1}$\\ \hline
  $K_{m5}$ & 15 & 7.6 & nM & $1.98^{-1}$\\ \hline
   $V_{6}$ & 0.75 & 0.77 & nM/s & $1.02$\\ \hline
  $K_{m6}$ & 15 & 13.9 & nM & $1.08^{-1}$\\ \hline
   $k_{7}$ & 0.001 & $5.2\e{-4}$ & 1/(s nM) & $1.93^{-1}$\\ \hline
   $k_{8}$ & 0.001 & $7.9\e{-4}$ & 1/(s nM) & $1.27^{-1}$\\ \hline
   $V_{9}$ & 0.5 & 0.49 & nM/s & $1.02^{-1}$\\ \hline
  $K_{m9}$ & 15 & 15.1 & nM & $1.01$\\ \hline
  $V_{10}$ & 0.5 & 0.51 & nM/s & $1.01$\\ \hline
 $K_{m10}$ & 15 & 15.4 & nM & $1.03$\\ \hline
  $x_{1t}$ & 100 & 100.2 & nM & $1.00$\\ \hline
  $x_{2t}$ & 300 & 300.2 & nM & $1.00$\\ \hline
  $x_{3t}$ & 300 & 304.4 & nM & $1.01$\\ \hline
\end{tabular} 
\label{tab:mapk-parameters}
\end{table}

The graph of $G(\xi_2,j\omega)$ is shown in Figure~\ref{fig:mapk-nyquist}. For the new parameters $p_2$,
the graph now encircles the point 1. By Theorem~\ref{theo:topological-equivalence}, we see
that the equilibria $\bar x_1$ and $\bar x_2$ are not topologically equivalent. Indeed,
$\bar x_2$ is unstable and the system converges to a limit cycle for parameters
$p_2$. The time course of these oscillations is plotted in Figure~\ref{fig:mapk-solutions-p0}.

\begin{figure}[htb]
\begin{center}
\includegraphics[width=0.65\linewidth]{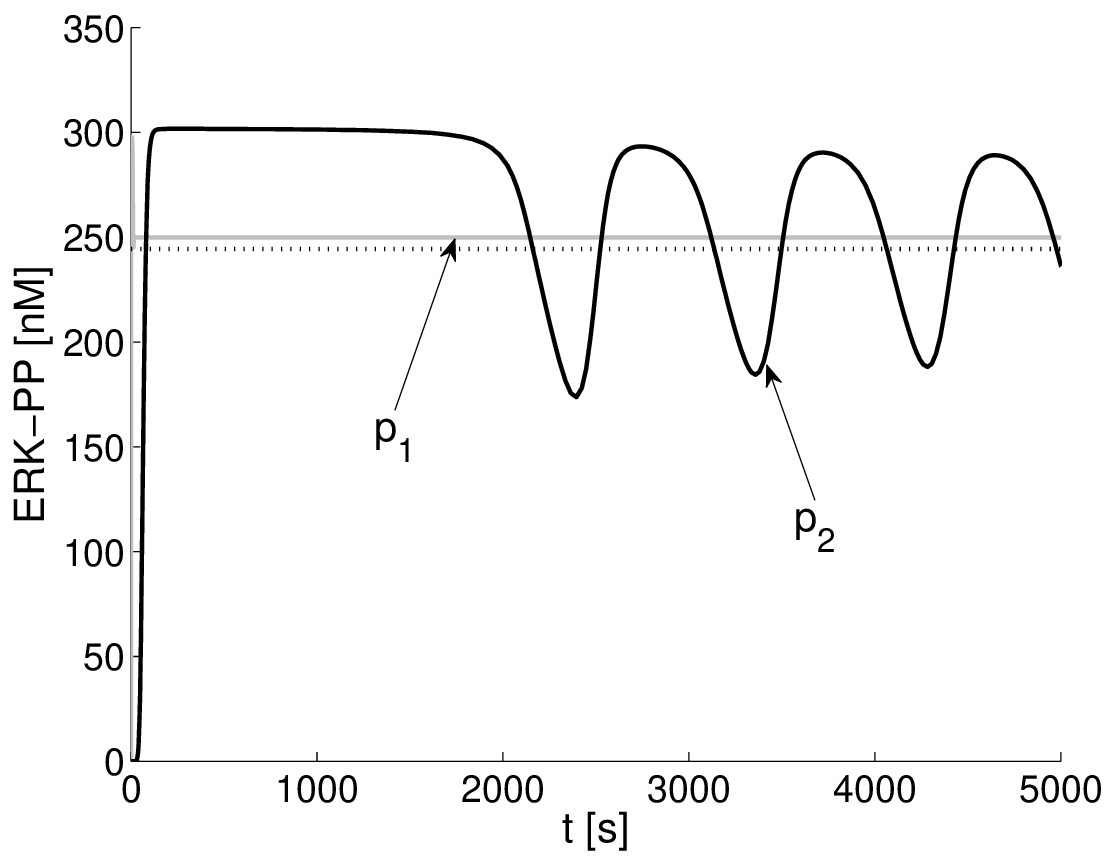}
\end{center}
\caption{Convergence to steady state for parameters $p_1$ (grey line) and sustained
oscillations for parameters $p_2$ (black line). The oscillations coexist with the unstable
equilibrium $\bar x_2$ (dotted line).}
\label{fig:mapk-solutions-p0}
\end{figure} 

\begin{figure}
\begin{center}
\includegraphics[width=\linewidth]{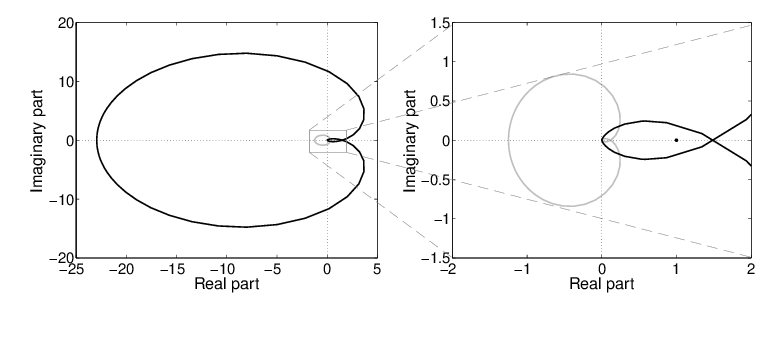}
\end{center}
\caption{Nyquist plots of open--loop MAPK model for parameters $p_1$ (grey line) and $p_2$ (black line).}
\label{fig:mapk-nyquist}
\end{figure}

In conclusion, our method is able to compute parameters which render the corresponding
equilibrium unstable and thus lead to the emergence of sustained oscillations
in the treated example. 
About half of the parameters are
varied by a non--negligible amount, but all variations are within the physiological range, the
highest variation being a factor of about $15$ in the $K_m$-value of one reaction.
It is also worth mentioning that the concentration values in the equilibrium did not change
significantly, although this should not be of physiological relevance for
the unstable equilibrium.

\section{Conclusions}

The loop breaking concept is introduced as a theoretical tool
to analyse complex behaviour in ODE systems as frequently encountered in
mathematical biology. Based on this tool, we present results on topological
equivalence of equilibria in systems with high-dimensional parameter spaces
and on the existence of critical parameters, for which stability properties of equilibria
may change. In addition, an algorithm is given to systematically search for
critical parameters. Using an ODE model for a MAPK cascade, we show that
the algorithm can be used to efficiently search for parameter values leading to limit
cycle oscillations in the system.

Non-uniqueness of critical parameter values is a problem that is inherent to this
kind of analysis. If the dimension of the parameter space is higher than the codimension
of the bifurcation, then there will be a submanifold of bifurcation points in the
parameter space. Our algorithm computes one of these points. Starting from a critical
point thus found, one can then use continuation methods to further
explore the structure of the set of critical points.

Another possibility for further studies would be to search for a bifurcation which is locally closest to
some reference parameter values. A method for this has been presented by \cite{Dobson2003}.
The method requires a bifurcation point where the search is started, 
and we expect our algorithm to give a starting point
which is better suited for the method discussed in \citep{Dobson2003} than a bifurcation 
search along a random line in parameter space.

The biological example we study in Section~\ref{sec:application} is simple in that
it contains only one feedback loop.
For systems with a single feedback loop, the results of the proposed analysis
method are independent of how the loop breaking point is chosen.
Of course many biological systems contain several intertwined feedback loops.
Then, the choice of the loop breaking point needs more attention, because
the results in general depend on this choice.
In our experience, it is often beneficial to try to break several feedback loops
at once.
Also a comparison of different loop breaking points is usually helpful and could
give hints to the role played by individual loops in the dynamical behaviour of
the system.

\section*{Acknowledgements}
We thank Jung-Su Kim and Madalena Chaves for helpful comments on a previous
version of the paper.
SW acknowledges funding by the Stuttgart Research Center for Simulation
Technology through the project ``Dynamical behaviour of complex systems''.

\section*{Appendix}
\addtocounter{section}{1}

\subsection*{Proof of Lemma~\ref{lem:eigenvalues}}
For simplicity of notation, we drop the dependence on $\xi$ of matrices $A$, $B$ and $C$.
By Schur's lemma, we have
\begin{equation*}
\begin{aligned}
\det(sI-A_{cl}) = \det(sI-A-BC) = \det\matrix{sI-A&-B\\-C&1}.
\end{aligned}
\end{equation*}
Let $(sI-A)_{-i}\in\Real^{(n-1)\times n}$ denote the matrix $(sI-A)$ with the
$i$-th row deleted. Then, by cofactor expansion,
\begin{equation*}
\begin{aligned}
\det\matrix{sI-A&-B\\-C&1} = 1\cdot\det(sI-A)-\sum_{i=1}^n (-1)^{n+1+i} b_i \det\matrix{(sI-A)_{-i}\\-C}.
\end{aligned}
\end{equation*}
In the same way,
\begin{equation*}
\begin{aligned}
\det\matrix{sI-A&-B\\-C&0} = - \sum_{i=1}^n (-1)^{n+1+i} b_i \det\matrix{(sI-A)_{-i}\\-C}
\end{aligned}
\end{equation*}
and with Prop.\ \ref{prop:closed-loop-jacobian} it follows that
\begin{equation}
\label{eq:closed-loop-determinant}
\begin{aligned}
\det(sI-A_{cl}) = \det(sI-A)- \det\matrix{sI-A&-B\\C&0}.
\end{aligned}
\end{equation}

$s_0$ is an eigenvalue of $A_{cl}$ if and only
if $\det(s_0 I-A_{cl}) = 0$.
For condition \textit{(i)}, we have $\det(s_0 I-A)\neq 0$,
and thus the equation
\begin{equation*}
\begin{aligned}
\frac{\det\matrix{s_0 I-A&-B\\C&0}}{\det(s_0 I-A)} = 1
\end{aligned}
\end{equation*}
is equivalent to $s_0$ being an eigenvalue of $A_{cl}$. The claim
then follows from \eqref{eq:transfer-function}.

The other case where $\det(s_0 I-A)= 0$ is considered in condition \textit{(ii)},
and \eqref{eq:closed-loop-determinant} can be used directly to prove the claim.

\subsection*{Proof of Proposition~\ref{prop:minimality}}
Note that $\alpha$ is constant over $\mathcal{M}$ due to assumption (A1).
Consider the transfer function $G(\xi,s)$ for a constant $\xi\in\mathcal{M}$.
For ease of notation, we drop $\xi$ in the transfer function in the following.

It is well known from linear control theory that the argument of $G(j\omega)$ changes
by $\alpha\pi$ when varying $\omega$ from $-\infty$ to $\infty$
\cite{DAzzoHou1975}:
\begin{equation*}
\begin{aligned}
\vert \arg G(j\infty)- \arg G(-j\infty)\vert = \alpha\pi.
\end{aligned}
\end{equation*}
The symmetry $G(j\omega)=\overline{G(-j\omega)}$ implies that $\arg G(j\infty) = - \arg G(-j\infty)$.
From these two facts, it follows that the argument of $G(j\omega)$ spans 
the open interval $I_\alpha = (-\frac{\alpha\pi}{2},\frac{\alpha\pi}{2})$
for $\omega\in(-\infty,\infty)$.

Moreover, by Definition~\ref{def:crit-freq} the condition
$\omega_c\in\Omega_c(\xi)$ is equivalent to
\begin{equation*}
\begin{aligned}
\arg G(j\omega_c) = k\pi,\qquad k\in\mathbb{Z}.
\end{aligned}
\end{equation*}

If $\alpha$ is even, the claim follows directly, since there are $\alpha-1$ different
integer values for $k$ such that $k\pi$ is inside the interval $I_\alpha$.
This corresponds directly to having $\alpha-1$ or more critical frequencies.

In the other case, if $\alpha$ is odd, some additional reasoning is needed to prove the proposition.
In this case one has
\begin{equation*}
\begin{aligned}
I_\alpha = \left( -\frac{2m+1}{2} \pi, \frac{2m+1}{2} \pi\right), \qquad m\in\mathbb{Z}.
\end{aligned}
\end{equation*}
Thus the borders of the interval $I_\alpha$ are not at integer multiples of $\pi$,
which implies that in this case there are $\alpha$ different integer values for
$k$ such that $k\pi$ is in $I_\alpha$, corresponding to at least $\alpha$ critical frequencies.

\subsection*{Proof of Theorem~\ref{theo:topological-equivalence}}
The proof for Theorem~\ref{theo:topological-equivalence} uses the argument principle from
complex analysis, which is repeated here for completeness \citep{WhittakerWat1965}.

\paragraph*{\rmfamily\upshape Theorem (The argument principle):}
\begin{itshape}
Let $f$ be a meromorphic function on the domain $D\subset\mathbb{C}$ and
$\Gamma$ a simply closed curve in $D$ such that $f$ does not have a
zero or pole on $\Gamma$. The winding number $wn(f(\Gamma),0)$ of the image
of $\Gamma$ under $f$ around the origin is given by
\begin{equation*}
\begin{aligned}
wn(f(\Gamma),0) = z_f - p_f,
\end{aligned}
\end{equation*}
where $z_f$ ($p_f$) is the number of zeros (poles) of $f$ in the interior of
the curve $\Gamma$, counted according to their algebraic multiplicities.
\end{itshape}

Note that the winding number is counted in the counter-clockwise direction.

As typically done in linear control theory, we will generally use the imaginary axis
for $\Gamma$, also called the Nyquist curve.
This can be seen as a closed curve by first taking only the interval $[-jR,jR]$ and
the half circle with radius $R$ in the right half plane, and second letting $R\rightarrow\infty$.
Thus the interior of $\Gamma$ is the right half plane.

We will first proof the intermediate results given in 
Lemmas~\ref{lem:ordered-crit-freq}--\ref{lem:same-winding-number},
before proving Theorem~\ref{theo:topological-equivalence}.
\begin{proof}\emph{(Lemma~\ref{lem:ordered-crit-freq})}
If $\Omega_c(\xi)$ is minimal, then there is exactly one $\omega_c\in\Omega_c(\xi)$ such
that $\arg G(\xi,j\omega_c)=k\pi$ for each $k\in\mathbb{Z}$ with $k\pi\in I_\alpha$ (where
$I_\alpha$ is the interval defined in the proof of Proposition~\ref{prop:minimality}).
This implies that \[\left\vert \arg G(\xi,j\omega_c^i(\xi)) - \arg G(\xi,j\omega_c^{i-1}(\xi))\right\vert = \pi\]
and thus $G(\xi,j\omega_c^i(\xi)) G(\xi,j\omega_c^{i-1}(\xi)) < 0$.
\end{proof}

\begin{proof}\emph{(Lemma~\ref{lem:winding-number})}
Note that the loop breaking \eqref{eq:loop-breaking} assures that $G(\xi,j\infty)=G(\xi,-j\infty)=0$.
Considering also Lemma~\ref{lem:ordered-crit-freq}, it follows that every cut
of $G(\xi_i,\Gamma)$ to the right of the point $1$ is preceded and followed by a cut
of $G(\xi_i,\Gamma)$ with the negative real axis. 
Thus each cut to the right of the point $1$ corresponds to one winding of $G(\xi_1,\Gamma)$
around the point $1$.
Moreover, Lemma~\ref{lem:ordered-crit-freq} assures that these windings all
have the same direction and thus several windings cannot cancel in the total winding number.
\end{proof}

\begin{proof}\emph{(Lemma~\ref{lem:same-winding-number})}
From assumption (A1), the transfer functions $G(\xi_1,\cdot)$ and
$G(\xi_2,\cdot)$ have the same number of zeros and poles in the left and right half
plane.
Thus for the phase differences we have
\begin{equation*}
\begin{aligned}
\arg G(\xi_1,j\infty) - \arg G(\xi_1,-j\infty)=\arg G(\xi_2,j\infty) - \arg G(\xi_2,-j\infty).
\end{aligned}
\end{equation*}
This implies that the winding numbers $wn(G(\xi_1,\Gamma),1)$ and $wn(G(\xi_2,\Gamma),1)$
have the same sign and with Lemma~\ref{lem:winding-number} we conclude
\begin{equation*}
\begin{aligned}
\big\vert wn(G(\xi_1,\Gamma),1) - wn(G(\xi_2,\Gamma),1) \big\vert &= 
\big\vert \vert wn(G(\xi_1,\Gamma),1)\vert - \vert wn(G(\xi_2,\Gamma),1)\vert \big\vert \\ &= 
\big\vert \beta(\xi_1) - \beta(\xi_2) \big\vert.
\end{aligned}
\end{equation*}
\end{proof}

We are now ready to give the proof of Theorem~\ref{theo:topological-equivalence}.
\begin{proof}(\emph{Theorem~\ref{theo:topological-equivalence}})
By Definition~\ref{def:topological-equivalence}, topological equivalence of
$\xi_1$ and $\xi_2$ is equivalent to the condition that the matrices $A_{cl}(\xi_1)$
and $A_{cl}(\xi_2)$ have the same number of eigenvalues with positive real part.

From the proof of Lemma~\ref{lem:eigenvalues}, we know that
\begin{equation*}
\begin{aligned}
G(\xi,s)-1 = \frac{\det (sI-A_{cl}(\xi))}{\det (sI-A(\xi))}.
\end{aligned}
\end{equation*}
Using the argument principle, it follows that
\begin{equation*}
\begin{aligned}
wn(G(\xi,\Gamma),1) = n_{cl}(\xi) - n_{ol}(\xi),
\end{aligned}
\end{equation*}
where $n_{cl}(\xi)$ ($n_{ol}(\xi)$) is the number of eigenvalues of $A_{cl}(\xi)$
($A(\xi)$) with positive real part.
By assumption, $n_{ol}(\xi_1)=n_{ol}(\xi_2)$ and thus $\xi_1$ and $\xi_2$
are topologically equivalent if and only if
\begin{equation*}
\begin{aligned}
wn(G(\xi_1,\Gamma),1) = wn(G(\xi_2,\Gamma),1).
\end{aligned}
\end{equation*}
The claim of the theorem then follows from Lemma~\ref{lem:same-winding-number}.
\end{proof}

\section*{Biographical information}

\begin{minipage}[c]{0.2\textwidth}
\includegraphics[width=0.9\linewidth]{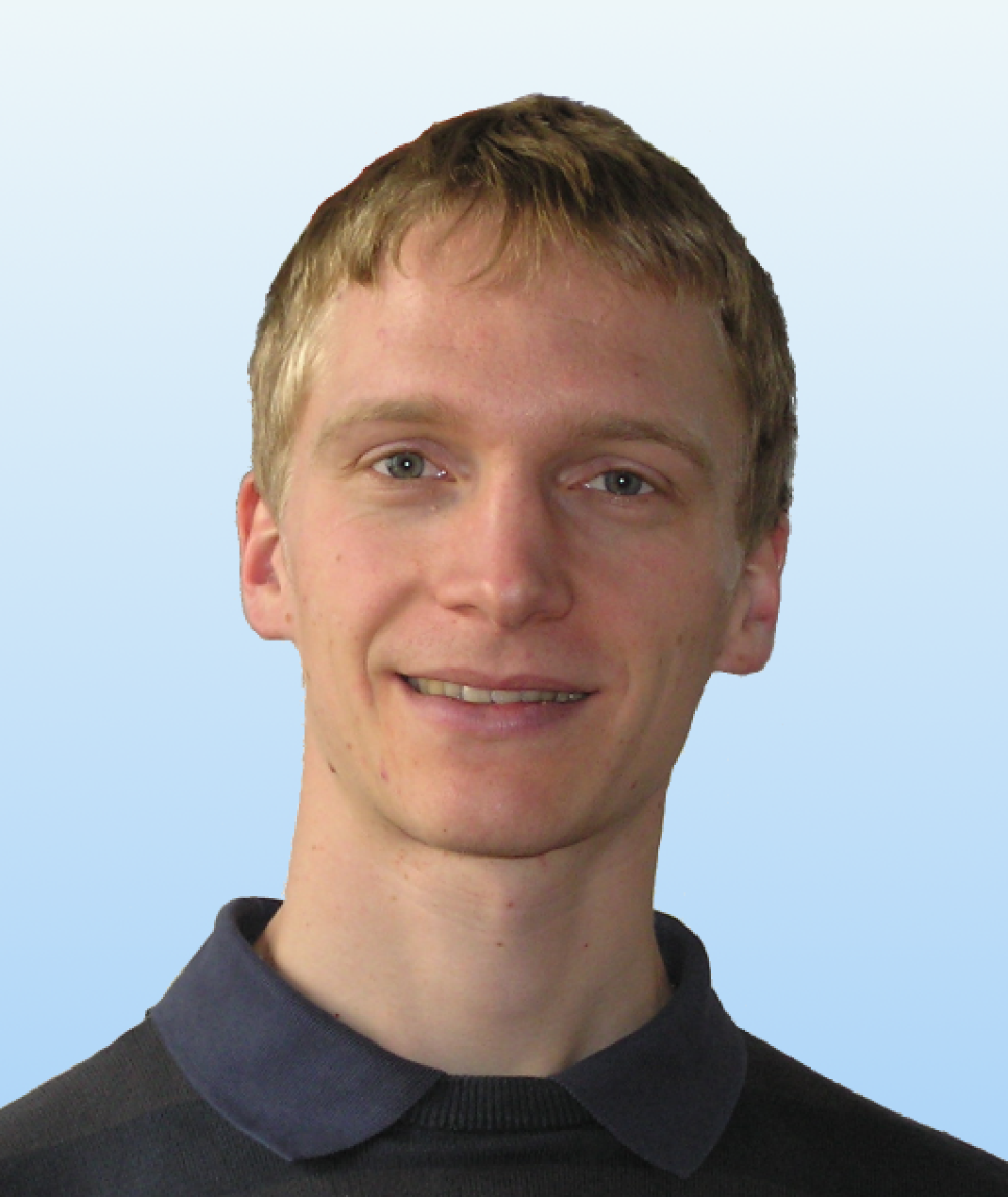}
\end{minipage}
\begin{minipage}[c]{0.75\textwidth}
\textbf{Steffen Waldherr} is research assistant at the Institute
for Systems Theory and Automatic Control at the University of
Stuttgart, Germany.
His main research interests are uncertainty analysis of biochemical
reaction networks, non-linear dynamics, and modelling
of biochemical signal transduction pathways.
\end{minipage}\\[0.5cm]
\begin{minipage}[c]{0.2\textwidth}
\includegraphics[width=0.9\linewidth]{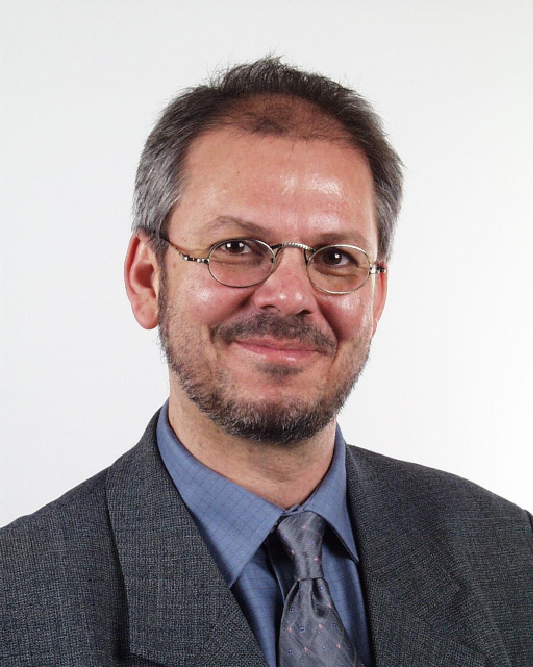}
\end{minipage}
\begin{minipage}[c]{0.75\textwidth}
\textbf{Frank Allgöwer} is director of the
Institute for Systems Theory and Automatic Control at the 
University of Stuttgart, Germany. 
His main research interests are nonlinear, robust, and predictive
control, identification, and applications in process engineering, 
systems biology, mechatronics, and nanotechnology.
\end{minipage}

\end{document}